\documentclass[11pt,reqno]{amsart}

\usepackage{amsmath,amssymb,amsthm,amsfonts,latexsym}
\usepackage{stmaryrd}

\usepackage{mathrsfs}

\usepackage{relsize}
\usepackage{comment}

\usepackage{graphicx,subfigure}

\usepackage[square,numbers]{natbib}

\numberwithin{equation}{section}

\usepackage[colorlinks=true,urlcolor=blue,
citecolor=red,linkcolor=blue,linktocpage,pdfpagelabels,
bookmarksnumbered,bookmarksopen,backref=page]{hyperref}

\usepackage{fullpage}

\usepackage{ifthen} 

\provideboolean{shownotes} 
\setboolean{shownotes}{true} 
\newcommand{\margnote}[1]{
\ifthenelse{\boolean{shownotes}}%
{\marginpar{\raggedright\tiny\texttt{#1}}}%
{}%
}

\newcommand{\hole}[1]{
\ifthenelse{\boolean{shownotes}}%
{\begin{center} \fbox{ \rule {.25cm}{0cm}
\rule[-.1cm]{0cm}{.4cm} \parbox{.85\textwidth}{\begin{center}
\texttt{#1}\end{center}} \rule {.25cm}{0cm}}\end{center}}
{}
}

\newcommand{\zer}{\mathbf{0}}
\newcommand{\R}{{\mathbb R}}
\newcommand{\M}{{\mathbb M}}
\newcommand{\Z}{{\mathbb Z}}
\newcommand{\C}{{\mathbb C}}

\newcommand{\Id}{{\mathbb I}_d}

\newcommand{\Idc}{{\mathbb I}_5}
\newcommand{\Idfo}{{\mathbb I}_4}

\newcommand{\cA}{{\mathcal{A}}}

\newcommand{\cR}{{\mathcal{R}}}

\newcommand{\cS}{{\mathcal{S}}}

\newcommand{\wun}{\mathbf{\hat{w}_1}}
\newcommand{\wdo}{\mathbf{\hat{w}_2}}

\newcommand{\ii}{\mathrm{i}}

\renewcommand{\Re}{\mbox{\rm Re}\,}
\renewcommand{\Im}{\mbox{\rm Im}\,}

\theoremstyle{plain}

\newtheorem{lemma}{Lemma}[section]

\newtheorem{theo}[lemma]{Theorem}

\theoremstyle{definition}

\newtheorem{remark}[lemma]{Remark}

\theoremstyle{remark}

\usepackage[foot]{amsaddr} 

\title{The secular equation for elastic surface waves under non standard boundary conditions of impedance type: A perspective from linear algebra}

\author[F. Vallejo]{*Fabio Andres Vallejo Narvaez}
 
\address{*Instituto de Investigaciones en Matem\'aticas Aplicadas y en Sistemas\\Universidad Nacional 
Aut\'onoma de M\'exico\\Circuito Escolar s/n, Ciudad Universitaria, Ciudad de M\'{e}xico C.P. 04510 (Mexico).}

\email{fabioval@ciencias.unam.mx}


\begin{document}

\maketitle

\begin{abstract}
The study of elastic surface waves under impedance boundary conditions has become an intensive field of research due to their potential  to model a wide range of problems. However, even when the secular equation, which provides the speed of the surface wave,  can be explicitly derived, the analysis is limited to specific cases due to its cumbersome final expression. In this work, we present an alternative method  based on linear algebra tools, to deal with the secular equation for surface waves in an isotropic elastic half-space subjected to non-standard boundary conditions of impedance type. They are defined by proportional relationships between both the stress and velocity components at the surface, with complex proportional ratios. Our analysis shows that the associated secular equation does not vanish in the upper complex half-plane including the real axis.  Interestingly, the full impedance boundary conditions proposed by Godoy et al. [Wave Motion 49 (2012), 585-594] arise as a particular limit case. 
 An approximation technique is introduced, in order to extend the analysis from the original problem to Godoy's impedance boundary conditions. As a result, it is  shows that the secular equation with full Godoy's impedance boundary condition does not vanish outside the real axis. This is a crucial property for the well-posedness of the boundary value problem of partial differential equations, and thus crucial for the model to explain surface wave propagation. Due to the cumbersome secular equation, this property has been verified only for particular cases of the impedance boundary condition, namely the stress-free boundary condition (zero impedance) and when either one of the impedance parameter is set to zero (normal and tangential impedance cases).

\end{abstract}


\setcounter{tocdepth}{1}
\tableofcontents


\section{Introduction}

Surface waves and their applications have been a central topic in a wide range of scientific fields, such as acoustics, the telecommunications industry,  material science,  and  notably in seismology, due to their potential to explain most of the damage and destruction during an earthquake. The best known surface waves are Rayleigh waves, which propagates along the free surface of an elastic medium, with an amplitude that decreases  exponentially  with the depth. The simplest  Rayleigh wave occurs along the surface of a homogeneous isotropic half-space subjected to the classical stress-free boundary condition and it was first described by Lord Rayleigh in his seminal work \cite{Raylei} from 1885. This work was the beginning of a lot of investigations on Rayleigh wave propagation on general anisotropic elastic half-space, where the  stres-free boundary condition constitutes the main paradigm (see, e.g., \cite{Raylei,Achen75,NakaG91,BarLoth1985,PhamOg04,Ting2002}). \\

Recently, there has been an increasing interest in surface wave propagation along elastic solids subjected to non-standard boundary condition of impedance type, which are defined by linear relations between the unknown function and its derivatives. 
Although they are commonly used in electromagnetism \cite{senior1960,STUP2011,Owen13} and acoustics \cite{Antipov2002,Zakharov2006,YlOijala2006,Qin2012}, they also have proven to be effective in modeling specific problems in linear elasticity. For instance,  in the study of surface wave propagation on a half-space coated by a tiny layer on the surface,  
 the impedance boundary conditions can be used to simulate the effects of the tiny layer without directly considering the layer itself (see \cite{Tier1,bovik1,Vinh2012,Vinh2014ah,Dai1}).  In seismology, Malischewsky showed the potential of impedance boundary condition to model seismic wave propagation along discontinuities \cite{Masky1,Masky2}. They can also be used to describe interfaces between solids  under certain conditions  (e.g. \cite{Murty1,Diaz2015,Durn2000}). In the mathematical framework of hyperbolic partial differential equations defined on the half-space, impedance boundary conditions are interesting because, in contrast to the usual stress-free boundary condition, they might lead a wide range of boundary value problems of hyperbolic partial differential equations (PDEs), ranging from  well-posed problems, for which the existence of a unique solution is guaranteed, to PDE problems with Hadamard inestabilities, where the existence or uniqueness of a solution fail to hold. See, for instance, \cite{BS,Ser5,BRSZ,serr2} for the linear case and \cite{RamFa1,Pl2,FrP1} for application to non-linear problems, and the references therein. In this context, problems involving surface wave propagation emerge as transition problems. Surface waves propagating in presence of impedance boundary conditions are clearly of practical as well as theoretical interest.\\

The central problem in the study of surface waves under non-standard boundary conditions is determining their existence and uniqueness, better known as the surface wave analysis. For  general anisotropy elastic half-spaces subjected to the standard stress-free boundary condition, the surface wave analysis can be performed  by means of several methods such as the polarization vector method, matrix impedance and the so-called Stroh formalism (see \cite{Lothe1,MalAji1,Ting1,BarLoth1985} and in the references therein). 
However, as noted by Giang and Vinh \cite{Pham21}, some of these methods cannot be  directly extended to the case of impedance boundary conditions. In such cases, the classical approach is employed in which the existence of a surface  wave is guaranteed by the existence of a unique real zero (in the subsonic range) of the secular equation. This is a non-linear algebraic equation that results impractical to solve analytically, even for the simplest configuration, namely isotropic solid with stress-free boundary condition (for an abridged list of references, see \cite{Godoy1,HaRiv62,Rahman1995,Masky2000,Masky04,VinhOg2004,Rahman06,Li2006}). This real root of the secular equation  is precisely the speed of the surface wave.\\ 

When general impedance boundary conditions are considered, the analysis of the associated secular equation may be quite challenging. Malischewsky \cite{Masky1} provided a reduced form of the secular equation for Tiersten's impedance boundary conditions, in terms of impedance parameters depending on the frequency and the material constants. In order to obtain general results about the existence of surface waves, Godoy et al. \cite{Godoy1} assumed  that the impedance parameters are proportional to the frequency (see Equation \eqref{imped121}). Under this assumption, the secular equation becomes independent of the frequency. In other words, the  impedance boundary conditions proposed by Godoy et al.  generalize the stress-free boundary condition in the non-dispersive regime. 
 However, due to the cumbersome  secular equation, Godoy et al. \cite{Godoy1} restrict themselves to the case of tangential impedance boundary condition (see Equation \eqref{imped121} with $Z_1\in\R$, $Z_2=0$) and prove that a surface wave is always possible for arbitrary values of the  tangential impedance parameter ($Z_1$). In a further work, via ingenious algebro-analytical manipulations based on Cauchy integrals from complex analysis,  Vinh and Nguyen \cite{Vinh17} were able to derive an exact analytical formula for the phase velocity of the surface wave described by Godoy et al. \cite{Godoy1}. These waves are often termed ``Rayleigh waves with impedance boundary conditions''. Recently, Giang and  Vinh \cite{Pham21} followed the same approach to study the case of normal impedance boundary condition (see Equation \eqref{imped121} with $Z_1=0$, $Z_2\in\R$). The authors found that, in contrast to the tangential case \cite{Godoy1}, there exist values of the parameters (Lam\'e constants and $Z_2$) for which surface waves are not possible. Conversely, for all parameter values  for which the existence and uniqueness of the surface wave is guaranteed, an exact formula for its velocity was provided.  However, as pointed out in  \cite{Pham21}, the case of normal impedance boundary condition demanded more technical details compared to  its tangential counterpart when applying the complex function method. This suggest that this method might be hard to apply to the general case with both tangential and normal impedance parameters non-zero ($Z_1,Z_2\in\R$). As far as we know, the existence of surface waves in the general case is still an open problem. Alternative methods are therefore highly desirable to analyze the secular equation in more intricate scenarios.\\

Although in this work the existence of surface waves of Rayleigh type is not stablished, we concentrate in another essential feature: the behavior of the secular equation off the real axis in the complex plane. In the mathematical theory of hyperbolic PDEs, to investigate the well-posedness property, the linear equations of elastodynamics are frequently  written as a first order hyperbolic system of PDEs, where the boundary conditions of impedance type take the form of linear relations among the components of the unknown vector. In this framework, it is known that roots of the secular equation along the upper complex half-plane  lead to Hadamard instabilities of the associated boundary value problem of PDEs (see \cite{BS,Ser5,serr2}). This means that the existence and /or uniqueness of a solution for the PDE problem fail to hold, making the surface wave analysis  meaningless \cite{BS}. This fact it is not completely unknown in the context of linear elasticity theory.  Hayes and Rivlin \cite{HaRiv62} showed that complex roots of the stress-free secular equation (if they exist) are associated with physically inadmissible displacement fields. Achenbach \cite{Achen75}, via the principle argument from complex analysis, proved that for appropriate Lam\'e constants (see Equation \eqref{lame44}), the secular equation for the strees-free case (see Equation \eqref{finsec12} with $\gamma_1=\gamma_2=0$) does not vanish outside the real axis. 
 Obviously, the use of the argument principle to deal with the secular equation with full impedance boundary conditions investigated in \cite{Godoy1,Vinh17,Pham21}  is impractical due to its cumbersome final expression. For the particular cases of tangential and normal impedance boundary conditions, the absence of complex roots off the real axis trivially follows from the analysis performed in the aforementioned works \cite{Vinh17,Pham21}.   
\\

The main purpouse of this paper is to prove that the secular equation with full impedance boundary conditions (proposed by Godoy et al.) does not have complex roots off the real axis for arbitrary real impedance parameters ($Z_1,Z_2\in\R$). To do this, we present an alternative approach based on linear algebraic tools that proves to be useful for studying the secular equation off the real axis and even along the entire complex plane in certain cases.
  This approach takes advantage  of the fact that through a suitable change of variables, the elastodynamic equations for a compressible isotropic half-space  can be written as a first-order symmetric linear hyperbolic system of PDEs. Section \S \ref{syme45} is devoted to this purpose.  In section  \S \ref{imped44}, we introduce a perturbed boundary condition which prescribes both the stress and velocity components at the surface to be proportional (see Equation \eqref{bouny0}), where the proportional ratios are complex-valued  constants. The term "perturbed" attached to the boundary condition refers to the fact that when the real part of the complex valued ratios vanish, the full impedance boundary condition proposed by Godoy et al. \cite{Godoy1},  is retrieved. In this section, we also compute the associated secular equation. In section \S \ref{rayanaly} we apply the alternative method based on quadratic forms  and use a  result from the well-posedness theory of hyperbolic PDEs to show that the secular equation associated with the perturbed problem does not vanish neither the real axis nor along the upper complex half-plane (indicating the absence of Hadamard instabilities). 
 In section \S \ref{imped678}, we take advantage of this result to demonstrate the non-vanishing property off the real axis for the secular equation with full impedance boundary condition. That is, the property that Achenbach \cite{Achen75} proved for the stress-free boundary condition, remains valid for the general impedance boundary conditions proposed by Godoy et al. \cite{Godoy1}. 

\subsection*{Notation}
In this paper, lowercase bold letters denote column vectors and uppercase bold letters denote matrices. Given a matrix $\mathbf{M}$ (or a vector $\mathbf{v}$), its conjugate transpose is denoted by $\mathbf{M}^*=\overline{\mathbf{M}}^{\top}$.  We equip $\C^n$ with the Hermitian scalar product
$$\mathbf{v}^{*}\mathbf{w}=\sum\limits_{j=1}^{n}\overline{v}_j w_j,$$
where $v_j$ is the $j$-th component of $\mathbf{v}$. This scalar product leads to the norm:
$$\|\mathbf{v}\|=\sqrt{\mathbf{v}^*\mathbf{v}}.$$
The $d$-dimensional indentity matrix is written $\Id$ .

\section{Equations of motion}
\label{syme45}

Let us consider an isotropic elastic half-space with constant mass density $\rho$ occupying the domain $\{x_2\geq0\}$. We shall study planar motion in the $(x_1,x_2)$-plane, the displacement being independent of $x_3$. The components of the displacement satisfy
\[
u_j=u_j(x_1,x_2,t),\;\text{for}\;j=1,2,\;\text{and}\;u_3\equiv 0.
\] 

Thus, the constitutive isotropic equations characterized by  the symmetric stress tensor $\sigma$  has four relevant components related to the displacement gradients by
\begin{equation}\label{streslin}
\begin{aligned}
\sigma_{11}&=(\lambda+2\mu)u_{1,1}+\lambda u_{2,2},\\
\sigma_{12}=\sigma_{21}&=\mu\big(u_{1,2}+u_{2,1}\big),\\
\sigma_{22}&=(\lambda+2\mu)u_{2,2}+\lambda u_{1,1},
\end{aligned}
\end{equation}
where commas denotes differentiation with respect to spatial variables $x_i$ and  $\mu,\lambda$ are the standard  Lam\'e constants satisfying
\begin{equation}\label{lame44}
\mu>0,\quad \lambda+\mu>0.
\end{equation}
In terms of the Young's modulus $E$ and the Poisson's ratio $\nu$, \eqref{lame44} is equivalent to $E>0$ and $-1<\nu<0.5$ (see \cite{Achen75}). In absence of source terms, the equations of motion are of the form:
\begin{equation}\label{stssq1}
\begin{split}
\sigma_{11,1}+\sigma_{12,2}=\rho\ddot{u}_1,\\
\sigma_{12,1}+\sigma_{22,2}=\rho\ddot{u}_2.
\end{split}
\end{equation}
Adopting the velocity components $v_1=\dot{u}_1$, $v_2=\dot{u}_2$ and the stress components $\sigma_{11},\sigma_{12},\sigma_{22}$ as the dependent variables, the equations of motion \eqref{streslin}-\eqref{stssq1} can be written as a first order linear system of PDE which is known as the velocity-stress formulation (see \cite{viri86}). Indeed, by taking the time derivative of the stress components in \eqref{streslin} and applying the  Clairaut's theorem for mixed second order time and spatial derivatives, we can write the equation of motion \eqref{streslin}-\eqref{stssq1} in the form (\cite{viri86}):  
\begin{equation}\label{sys2or0}
\left\{\begin{split}
\rho\dfrac{\partial v_1}{\partial t}&=\dfrac{\partial \sigma_{11}}{\partial x_1}+\dfrac{\partial \sigma_{12}}{\partial x_2},\\
\rho\dfrac{\partial v_2}{\partial t}&=\dfrac{\partial \sigma_{12}}{\partial x_1}+\dfrac{\partial \sigma_{22}}{\partial x_2},\\
\dfrac{\partial \sigma_{11}}{\partial t}&=(\lambda+2\mu)\dfrac{\partial v_1}{\partial x_1}+\lambda\dfrac{\partial v_2}{\partial x_2},\\
\dfrac{\partial \sigma_{12}}{\partial t}&=\mu\Big(\dfrac{\partial v_1}{\partial x_2}+\dfrac{\partial v_2}{\partial x_1}\Big),\\
\dfrac{\partial \sigma_{22}}{\partial t}&=\lambda\dfrac{\partial v_1}{\partial x_1}+(\lambda+2\mu)\dfrac{\partial v_2}{\partial x_2}.\\
\end{split}\right. 
\end{equation}
In terms of the vector variable  $\mathbf{y}=(v_1,v_2,\sigma_{11},\sigma_{12},\sigma_{22})^\top$,  \eqref{sys2or0}  can be written in a compact  form as
\begin{equation}\label{sys2or00}
\dfrac{\partial \mathbf{y}}{\partial t}=\mathbf{A}_1\dfrac{\partial \mathbf{y}}{\partial x_1}+\mathbf{A}_2\dfrac{\partial \mathbf{y}}{\partial x_2}
\end{equation}
where $\mathbf{A}_1$, $\mathbf{A}_2$ are the constant matrices 
$$\mathbf{A}_1=\begin{pmatrix} 0 & 0 & 1/\rho & 0 &0 \\  0 & 0 & 0 & 1/\rho & 0 \\ \lambda+2\mu & 0 & 0 & 0 & 0 \\  0 & \mu & 0 & 0 & 0\\ \lambda & 0 & 0 & 0 & 0\end{pmatrix},\:\: \mathbf{A}_2=\begin{pmatrix} 0 & 0 & 0 & 1/\rho & 0  \\ 0 &  0 & 0 & 0 & 1/\rho  \\ 0 & \lambda & 0 & 0 & 0 \\  \mu & 0 & 0 & 0 & 0 \\ 0 & \lambda+2\mu  & 0 & 0 & 0 \end{pmatrix}.$$
This form of the isotropic elastodynamic equations is known as the velocity stress formulation and is frequently  used in seismology to implement  numerical methods (see, for instance, \cite{Jian99,viri86}). 
Our main results are based on the fact that the system \eqref{sys2or00} admits a symmetric respresentation  after a suitable change of variables (see \cite{Morando05,BS}).
Indeed,  change the variable $\mathbf{y}$ by the variable $\mathbf{w}=(w_1,w_2,w_3,w_4,w_5)^{\top}$ that is related to the former by $\mathbf{y}=\mathbf{C}\mathbf{w}$, where $\mathbf{C}$ is the constant  invertible matrix
$$\mathbf{C}=\begin{pmatrix} 0 & 0 & 0 &  \frac{1}{c_1} & 0  \\  0 & 0 & 0 & 0 & \frac{1}{c_1} \\ \frac{2\rho c_2\sqrt{c_1^2-c_2^2}}{c_1^2} & 0 &\rho\Big(1-\frac{2c_2^2}{c_1^2}\Big)  & 0 & 0 \\  0 & \frac{c_2\rho}{c_1} & 0 & 0 & 0\\ 0 & 0 & \rho & 0 & 0\end{pmatrix}$$
and $c_1:=\sqrt{(\lambda+2\mu)/\rho},\;c_2:=\sqrt{\mu/\rho}$ are the speed of bulk waves (preassure and shear, respectively). Substituting $\mathbf{y}=\mathbf{C}\mathbf{w}$ into \eqref{sys2or00} and simplifying gives

\begin{equation}\label{sys2or}
\dfrac{\partial \mathbf{w}}{\partial t}=\cS_1\dfrac{\partial \mathbf{w}}{\partial x_1}+\cS_2\dfrac{\partial \mathbf{w}}{\partial x_2},
\end{equation}
where $\cS_1=\mathbf{C}^{-1}\mathbf{A}_1\mathbf{C}$ and $\cS_2=\mathbf{C}^{-1}\mathbf{A}_2\mathbf{C}$ are the symmetric matrices given by
\begin{equation}\label{symatri}
\cS_1=\begin{pmatrix} 0 & 0 & 0 & \frac{2 c_2\sqrt{c_1^2-c_2^2}}{c_1} &0 \\  0 & 0 & 0 & 0 & c_2 \\ 0 & 0 & 0 & \frac{c_1^2-2c_2^2}{c_1} & 0 \\  \frac{2 c_2\sqrt{c_1^2-c_2^2}}{c_1} & 0 & \frac{c_1^2-2c_2^2}{c_1} & 0 & 0\\ 0 & c_2 & 0 & 0 & 0\end{pmatrix},\:\: \cS_2=\begin{pmatrix} 0 & 0 & 0 & 0 & 0  \\ 0 &  0 & 0 & c_2 & 0  \\ 0 & 0 & 0 & 0 & c_1 \\  0 & c_2 & 0 & 0 & 0 \\ 0 & 0  & c_1 & 0 & 0 \end{pmatrix}.
\end{equation}
Equation \eqref{sys2or} is refered as the symmetric first order version of the isotropic elastodynamic equations \eqref{stssq1}.

\section{Boundary conditions of impedance type and the secular equation}\label{imped44}
\subsection{Perturbed impedance boundary condition}
We assume that the surface $\{x_2=0\}$ is subjected to a boundary condition of the form:
\begin{equation}\label{bouny0}
\begin{aligned}
\sigma_{12}&+\gamma_1\dot{u}_1=0,\\
\sigma_{22}&+\gamma_2\dot{u}_2=0,
\end{aligned}\:\:\text{for}\:\: x_2=0,
\end{equation}
where $\gamma_1,\gamma_2$ are  complex constants given by
$$
\gamma_1=\varepsilon_1+\ii Z_1,\quad\gamma_2=\varepsilon_2+\ii Z_2.
$$
$Z_1,Z_2\in\R$ are the impedance parameters whose dimension is of stress/velocity \cite{Godoy1,Masky1} and $\varepsilon_1,\varepsilon_2$ are assumed to be negative reals (or $\Re\gamma_j<0$, $j=1,2$). We claim \eqref{bouny0} is a perturbed version of the full impedance boundary condition proposed by Godoy et al. \cite{Godoy1}. Indeed, assume that $\varepsilon_1,\varepsilon_2$ go to zero and that the displacement vector in \eqref{bouny0} depends harmonically on time through $e^{-\ii\omega t}$, that is $u_j=e^{-\ii\omega t}\hat{u}(x_1,x_2)$, $j=1,2$ (cf. \cite{Godoy1}). Substituting at \eqref{bouny0} gives the full impedance boundary condition (see, Equation (9) from \cite{Pham21})
\begin{equation}\label{imped121}
\left\{\begin{matrix}
\hat{\sigma}_{12}+\omega Z_1\hat{u}_1=0,\\
\hat{\sigma}_{22}+\omega Z_2\hat{u}_2=0,
\end{matrix}\right.  \:\:\: x_2=0.
\end{equation}
These impedance boundary conditions are tantamount to the Malischewsky ones if the  impedance parameters defined by him (see Equation 2 in \cite{Masky2}) are taken as:  $\xi_1=\omega Z_1$, $\xi_2=\omega Z_2$.  It is clear that setting $Z_1=Z_2=0$ leads to the classical stress-free boundary condition. When $Z_2=0$, we retrieve the tangential boundary condition investigated in \cite{Godoy1,Vinh17} and when $Z_1=0$ we obtain the normal impedance boundary condition studied  in \cite{Pham21}.



We are going to derive the secular equation for surface waves associated to the perturbed boundary condition (PBC) \eqref{bouny0} and analyze it in detail by a linear algebra approach.  Then,  we let $\varepsilon_1,\varepsilon_2$ tend to zero to obtain the secular equation with impedance boundary conditions. Finally, we analyze it in light of the behavior of its counterpart with PBD.


For the analysis to come, the symmetric first order version of the elastodynamic equation shall be fundamental, so we express the boundary condition \eqref{bouny0} in terms of the components of the vector $\mathbf{w}$. That is:
\begin{equation}\label{bouny}
\begin{aligned}
\rho c_2 w_2&+\gamma_1 w_4=0,\\
 \rho c_1 w_3&+\gamma_2 w_5=0,
\end{aligned}\:\:\text{for}\:\: x_2=0,
\end{equation}
which, by defining the column vector $\mathbf{w'}:=(w_2,w_3,w_4,w_5)^{\top}\in\C^4$, can be written in matrix form as:
\begin{equation}
\label{prefacto}
\begin{pmatrix}  c_2\rho & 0 & \gamma_1 & 0\\ 0 & c_1\rho & 0 &\gamma_2 \end{pmatrix}\mathbf{w'} =\begin{pmatrix} 0 \\ 0 \end{pmatrix},\:\: x_2=0.
\end{equation}

\subsection{Secular equation}
We derive in this section the secular equation for surface waves with PBC by using the first order symmetric version \eqref{sys2or}  of the elastodynamic equations and the matrix form of the boundary condition \eqref{prefacto}. The procedure is  the same as in  \cite{Achen75}, a surface wave of Rayleigh type propagating in the $x_1$-direction with velocity $c$,  and wave number $k>0$ has displacement vector given by
\begin{equation}\label{helch99}
\mathbf{w}=\textit{\LARGE e}^{-kb x_2}\textit{\LARGE e}^{k\ii(x_1-ct)}\mathbf{\hat{w}}
\end{equation}
where the constant vector $\mathbf{\hat{w}}\in\C^5$,  the velocity $c$ and the unknown $b$ have to be chosen  such that \eqref{helch99} satisfies both the differential equation \eqref{sys2or} and the boundary condition \eqref{bouny}. Moreover, the unknown $b$ must be chosen with positive real part in order to fulfill the decaying condition:
\begin{equation}\label{decay1}
\mathbf{w}=0\:\:\text{as}\:\: x_2\to+\infty.
\end{equation}

We now start to susbtitute  \eqref{helch99} into  \eqref{sys2or}. After some algebraic simplifications, we find that the vector $\mathbf{w}$ must solve the following linear homogeneous system in matrix form
\begin{equation}\label{eq9helch}
\big(ck\ii\Idc+k\ii \cS_1-kb \cS_2\big)\mathbf{\hat{w}}=0.
\end{equation}
Non-trivial solutions of the system above are necessary to have non-trivial solutions of the form \eqref{helch99}, so the determinant of the system \eqref{eq9helch} must vanish, that is 
\begin{equation}\label{detito}
\det\big(ck\ii\Idc+k\ii S_1-kb S_2\big)=0. 
\end{equation}
After algebraic manipulations, we find that this happens when $c=0$ (which is discarded), $b=\pm b_1$ and $b=\pm b_2$, where
\begin{equation}\label{const15}
{b}_1=\sqrt{1-\dfrac{c^2}{c_1^2}},\;\;\;\;\;{b}_2=\sqrt{1-\dfrac{c^2}{c_2^2}}.
\end{equation}
Observe that the square roots hereabove take complex values when $c>c_j$, $j=1,2$, so an exact meaning as a complex  functions is needed. For this, we assume that the square root in \eqref{const15} is the principal branch, which ensures $\Re b_1>0$ and $\Re b_2>0$ as $c$ varies on the whole complex plane. Thus, in order to fulfill the decaying condition outlined in  \eqref{decay1}, we select the solutions of \eqref{detito} with positive  sign, namely $b=b_j$, $j=1,2$. 
Now, solving the system \eqref{eq9helch} for each value $b=b_1$, $b=b_2$ (separately) shows that  the infinite set of solutions, for each value, are respectively spanned by the vectors
\begin{equation}\label{vecsray}
\begin{split}
\wun&=\Big(\dfrac{-2c_2k\sqrt{c_1^2-c_2^2}}{c_1},-2\ii c_2b_1k,\dfrac{c_2^2k(1+b_2^2)}{c_1},ck, b_1 ck\ii\Big)^{\top},\\
\wdo&=\Big(\dfrac{-2c_2k\sqrt{c_1^2-c_2^2}}{c_1},-c_2k\ii\Big(\frac{1}{b_2}+b_2\Big),\dfrac{2c_2^2k}{c_1},ck, \dfrac{ck\ii}{b_2}\Big)^{\top}.
\end{split}
\end{equation}
Replacing in \eqref{helch99}, we obtain two linear independent solutions  $\mathbf{w_1}$, $\mathbf{w_2}$ to \eqref{sys2or} given by
\begin{equation}\label{Raymod3q}
\mathbf{w_1}=\textit{\LARGE e}^{-kb_1 x_2}\textit{\LARGE e}^{k\ii(x_1-ct)}\wun,\;\;\;\;\;
\mathbf{w_2}=\textit{\LARGE e}^{-kb_2 x_2}\textit{\LARGE e}^{k\ii(x_1-ct)}\wdo.
\end{equation}
If we consider just one of these solutions, then there are values of the parameters $\gamma_1,\gamma_2$ for which the boundary condition does not hold. Indeed, if we take, for instance, any scalar multiple of the first mode $\mathbf{w_1}$ in \eqref{Raymod3q} and substitute it  into the boundary condition \eqref{bouny}, we obtain the following algebraic system of equations
\begin{equation}\label{bounfake}
\left\{\begin{matrix}
-2c_2^2\rho\ii k b_1+\gamma_1 c k=0,\\
\rho c_2^2k(1+b_2^2)+\gamma_2 b_1 ck\ii=0.
\end{matrix}\right. 
\end{equation}
Note that for any pair $\gamma_1,\gamma_2\in\C$ such that $\gamma_1=0$, the first equation in \eqref{bounfake} is not satisfied, provided that $c\neq c_1$.  Similarly, there are values of the  parameters for which the second mode in \eqref{Raymod3q} does not satisfy the boundary condition. 
Hence, for the sake of generality, we assume that a general surface wave solution of Rayleigh type to the system of PDEs \eqref{sys2or} is a linear combination of $\mathbf{w_1}$ and $\mathbf{w_2}$  in \eqref{Raymod3q}.  That is
\begin{equation}\label{moddef56}
\mathbf{w}=\big(A_1\textit{\LARGE e}^{-kb_1 x_2}\wun+A_2\textit{\LARGE e}^{-kb_2 x_2}\wdo\big)\textit{\LARGE e}^{k\ii(x_1-ct)}.
\end{equation}
Now, we have to find $A_1,A_2$ and $c$ such that \eqref{moddef56} satisfies the boundary condition. As before, substitute \eqref{moddef56} into \eqref{bouny} to conclude that the scalars $A_1$ and $A_2$ must solve the  homogeneous linear system 
\begin{equation}\label{ampliec0}
\begin{pmatrix}
-2\ii c_2^2\rho b_1k +ck\gamma_1 &-\ii c_2^2\rho k\Big(b_2+\frac{1}{b_2}\Big)+ ck\gamma_1\vspace{.3cm}\\
c_2^2\rho k(1+b_2^2)+ ck\gamma_2\ii b_1 & 2c_2^2\rho k +\dfrac{kc\gamma_2\ii}{b_2}
\end{pmatrix}\begin{pmatrix}A_1\\ A_2\end{pmatrix}=\begin{pmatrix} 0 \\ 0 \end{pmatrix}.
\end{equation}
If we make $\varepsilon_1=\varepsilon_2=0$ (namely, $\gamma_1=Z_1\ii$, $\gamma_2=Z_2\ii$) and multiply the first equation by $\ii$, we retrieve equation (11) from \cite{Pham21}, inasmuch as $\omega=ck$ and $c_2^2\rho(1+b_2^2)=(\lambda+2\mu)b_1^2-\lambda$.

Again, we need the system \eqref{ampliec0} to support more solutions than the trivial one, $A_1=A_2=0$, so the determinant of the system must vanish. This leads to the secular equation
\begin{equation}\label{finsec12}
\begin{split}
\cR(c;\gamma_1,\gamma_2):=\left(2-\dfrac{c^2}{c_2^2}\right)^2 &-4\sqrt{1-\dfrac{c^2}{c_2^2}}\sqrt{1-\dfrac{c^2}{c_1^2}}-\dfrac{c^3\ii}{\mu c_2^2}\left(\gamma_1\sqrt{1-\dfrac{c^2}{c_2^2}}+\gamma_2 \sqrt{1-\dfrac{ c^2}{c_1^2}}\right)\\
&+c^2\frac{\gamma_1\gamma_2}{\mu^2}\left(1-\sqrt{1-\dfrac{c^2}{c_2^2}}\sqrt{1-\dfrac{c^2}{c_1^2}}\right)=0.
\end{split}
\end{equation}
As expected, letting $\varepsilon_1,\varepsilon_2$ go to zero in \eqref{finsec12} leads to the secular equation with full impedance boundary condition as a function of the speed $c$
\begin{equation}\label{finsec50}
\begin{split}
\cR(c;Z_1\ii,Z_2\ii)&=\left(2-\dfrac{c^2}{c_2^2}\right)^2 -4\sqrt{1-\dfrac{c^2}{c_2^2}}\sqrt{1-\dfrac{c^2}{c_1^2}}+\dfrac{c^3}{\mu c_2^2}\left(Z_1\sqrt{1-\dfrac{c^2}{c_2^2}}+Z_2 \sqrt{1-\dfrac{ c^2}{c_1^2}}\right)\\ &+c^2\frac{Z_1 Z_2}{\mu^2}\left(\sqrt{1-\dfrac{c^2}{c_2^2}}\sqrt{1-\dfrac{c^2}{c_1^2}}-1\right)=0.
\end{split}
\end{equation}
Note that the secular equation hereabove is independent of the frequency, so Godoy's impedance boundary conditions generalizes the stress-free boundary condition in the non-dispersive regime. It is not hard to verify that in the variables $\delta_j=Z_j/\sqrt{\mu\rho}$, $x=c^2/c_2^2$, $\tau=c_2^2/c_1^2$, the secular equation above becomes Equation 12 in \cite{Pham21}.

In this fashion, the PDE boundary problem \eqref{sys2or}-\eqref{bouny0} supports a surface wave if there exists a unique real root of \eqref{finsec5} in the interval $(0,c_2)$ (subsonic range), corresponding to the speed of the wave (see, \cite{Vinh17,Godoy1,Pham21}). In this work, however, we concentrate on another  property regarding well-posedness of the PDE boundary problem \eqref{sys2or}-\eqref{prefacto}: the absence of roots of the secular equation in the upper complex half-plane $\{\Im c>0\}$.
Why are these roots forbidden?  Observe that a root $c$ with $\Im c>0$ of the secular equation leads to an associated solution of the form \eqref{helch99} that diverges (in norm) as $t\to +\infty$ (this is a wave of infinite energy).  One may argue that, in case they appear, these solutions should simply be omitted, given their lack of physical meaning. However the situation is not that easy.  It is well-known, in the  theory of hyperbolic PDEs, that these solutions of infinity energy cause the ill-posedness of the associated boundary value problem \eqref{sys2or}-\eqref{prefacto}. That is, the existence or uniqueness of the solution fail to hold once the data of the problem (L\'ame constants, source term, initial data) are prescribed. This is the case, for instance, when $\gamma_1=\gamma_2=0$  (stress-free boundary condition) and the Lam\'e constants are chosen inappropriately as $\mu>0$, $-\mu<\lambda+\mu\leq0$. The secular equation \eqref{finsec12} has  two positive roots that are meaningless, as the presence of an additional root within the region $\{\Im c>0\}$ implies the existence of a wave of infinite energy. Consequently, the problem is ill-posed under this range of parameters  (refer to Theorem 6.2 in \cite{Ser5} for details). The same situation arises when  $\gamma_1=\gamma_2=\gamma>0$ (in \eqref{bouny})   and the L\'ame parameters are set as in \eqref{lame44}. As in the previous example,  the secular equation \eqref{finsec12} has at least one root with $\Im c>0$ for each positive value of the  impedance parameter $\gamma$, implying the ill-posedness of the PDE problem (refer to Proposition 5.1 in \cite{BRSZ} for details).

\begin{remark}
In the literature of linear hyperbolic PDEs, the absence of roots of the secular equation within the region $\{\Im c>0\}$ is strongly related to the Kreiss Lopatinski\u{\i} condition (cf. \cite{serr2,BS}). This is a necessary condition for the  well-posedness of general first-order hyperbolic systems  with constant matrix coefficients endowed with a boundary condition in form of linear relations, like \eqref{bouny}. The problem under consideration \eqref{sys2or}-\eqref{prefacto} falls into this general class. Several hyperbolic problems in electromagnetism, fluids and gas dynamics are often written into systems of first-order equations  to investigate the well-posedness property in the light of Kreiss' theory. A detailed account of this theory and their numerous implications can be found in the monograph by Benzoni-Gavage and Serre \cite{BS}. 
\end{remark}
Since the well-posedness property must be satisfied for the model under consideration to appropriately simulate wave propagation, we can state the following
\begin{remark}\label{upper1}
A necessary condition for the problem \eqref{sys2or}-\eqref{prefacto} to support a surface wave of Rayleigh type   is the absence of roots of the secular equation along the upper complex half-plane $\{\Im c>0\}$. Observe that the function $\cR$ that defines the secular equation with impedance boundary conditions \eqref{finsec50} satisfies the following symmetry property:
\begin{equation}\label{unif89}
\cR(-c;Z_1\ii,Z_2\ii)=\cR(c;-Z_1\ii,-Z_2\ii).
\end{equation}
This property extends the symmetric property with respect to the origin  (in the complex plane) of the stress-free secular equation ($Z_1=Z_2=0$). Thus, the absence of roots of the secular equation in $\{\Im c>0\}$ for all $Z_1,Z_2\in\R$ equals the non existence of complex roots outside the real axis for all $Z_1,Z_2$.
\end{remark}

For the case of stress free boundary condition $Z_1=Z_2=0$,  Achenbach \cite{Achen75} verified,  via the argument  principle from complex analysis, that the associated secular equation does not have roots outside the real axis. This, in turn, implies the uniqueness of the Rayleigh wave. For the case of tangential ($Z_2=0$) and normal ($Z_1=0$) impedance boundary conditions, Vinh and Nguyen \cite{Vinh17} and Giang and  Vinh \cite{Pham21},  respectively, analyzed the secular equation \eqref{finsec50} across the entire complex plane  by means of the complex function method based on Cauchy integrals. For the tangential case, it was found that, in the variable  $x=c^2/c_2^2$, the secular equation has a unique real root on $(0,1)$ for each value of the impedance parameter $Z_1\in R$ which in particular implies the existence and uniqueness of a surface wave of  Rayleigh type. For the normal case, this property holds true iff the impedance $Z_2\in \R$ remains above a critical value depending on the material parameters. In both particular cases, the absence of roots for the corresponding secular equations off the real axis is a  trivial consequence of the results in those works. In Section \ref{imped678}, we extend this property to the general case with full impedance $Z_1,Z_2\in\R$. To do that, we first establish the property for  the secular equation with PBC \eqref{finsec12} and then extend it to the case of impedance boundary conditions by letting $\varepsilon_1,\varepsilon_2\to 0$.

\begin{remark}

It is worth mentioning that in contrast to the isotropic case, complex roots of the secular equation associated for instance to some linear  models  for a viscoelastic material are not necessarily physically inadmissible \cite{Sharma20v,Romeo2002}. 
The reason is that,  when written as a first order PDE like \eqref{sys2or00}, the matrix coefficients of the resulting system depends upon  the spatial variables and an additional non-homogeneous term appears, which accounts for the integral part in the constitutive relation that describes the viscoelastic behavior  (see, for instance, \cite{Hoop2018}).  This kind of system ranges into the class of hyperbolic systems with relaxation, and therefore  the cited theory of well-posedness for the symmetric first order version of the isotropic equations with constant matrix coefficients, is not yet valid  for the viscoelastic case.

\end{remark}

\section{The secular equation associated to the perturbed boundary condition (PBC)}\label{rayanaly}
In this section, we shall prove that the secular equation with PBC \eqref{finsec12} does not vanish, neither on the upper complex plane $\Im c>0$ nor along the real axis, for all boundary parameters such that $\Re\gamma_j=\varepsilon_j<0$,  $j=1,2$.  Note that $c=0$ is a trivial root of the secular equation which, however,  results spurious as we are interested in surface waves with a non-zero speed. Thus, we focus on the set 
$$\{\Im c\geq0\}^*=\{c\in\C:\Im c\geq0,\;c\neq 0\}.$$
We do not deal directly with the secular equation, given its intricacy expression.  
 Our alternative approach is based on the widely known fact from linear algebra that a square  homogeneous linear system has a unique solution (the trivial one) if and only if the determinant of the matrix of the system does not vanish. Since the  secular equation is precisely the determinant of the homogeneous square system \eqref{ampliec0} set to zero, then the non-vanishing property  along  $\{\Im c\geq0\}^*$ is equivalent to demonstrating that  $A_1=A_2=0$ is the unique solution of the system \eqref{ampliec0} for all $c\in \{\Im c\geq0\}^*$.   
 We shall further develop this idea by relying on the fact that \eqref{ampliec0} can be rewritten as
\begin{equation}
\label{facto1}
k \begin{pmatrix}  c_2\rho & 0 & \gamma_1 & 0\\ 0 & c_1\rho & 0 &\gamma_2 \end{pmatrix} \Big(A_1\wun'+A_2\wdo'\Big)=\begin{pmatrix} 0 \\ 0 \end{pmatrix},
\end{equation}
where $\wun',\wdo'\in\C^4$ denotes the vectors obtained from $\wun,\wdo$ in \eqref{vecsray} by eliminating  their first components, respectively. Straightforward matrix multiplication on \eqref{facto1} produces the matrix form of the  linear system \eqref{ampliec0}. This alternative form of the system arises when one substitute the surface wave \eqref{moddef56} into the boundary condition \eqref{prefacto}.  Vectors $\wun',\wdo'$ satisfy the following fundamental property:
\begin{lemma}
 The vectors $\wun',\wdo'$ are linear independent for all $c\neq0$.
\end{lemma}
\begin{proof}
 Indeed,  if we assume that the column vectors $\wun'$ and $\wdo'$ are parallel (linear dependent). That is, $\wun'=\alpha_0 \wdo'$, $\alpha_0\neq 0$; then from the last component of this equation we have $\alpha_0=b_1 b_2$; susbtituting into the second one produces $c=0$  which is a contradiction.
\end{proof}
According to the lemma above,  if we define
\begin{equation}\label{bomat1}
B_\gamma:=\begin{pmatrix}  c_2\rho & 0 & \gamma_1 & 0\\ 0 & c_1\rho & 0 &\gamma_2 \end{pmatrix},
\end{equation}
Equation \eqref{facto1} is clearly equivalent to looking for vectors $\mathbf{r}$ in the linear space $span\{\wun',\wdo'\}$ (the set of linear combinations) that solve the equation $B_\gamma \mathbf{r}=0$. Observe that, if the unique solution is the trivial vector $\mathbf{r}=0$ for all $c\in\{\Im c\geq0\}^*$,  then, the linear independence of $\wun',\wdo'$ implies that $A_1=A_2=0$ is the unique solution of the system \eqref{ampliec0}, which gives the non-vanishing property of the secular equation. The following lemma provides an alternative way to prove this fact.


\begin{lemma}\label{keylemma}
Suppose that for all $c \in\{\Im c\geq0\}^*$, there exists  a constant $\epsilon_0>0$ independent of $c$ such that
\begin{equation}\label{key1a}
\|B_\gamma \mathbf{r}\|_2\geq \epsilon_0\|\mathbf{r}\|_4, \:\text{for all}\:\mathbf{r}\in span\{\wun',\wdo'\},
\end{equation}
where $\|.\|_n$ denotes the usual norm in $\C^n$. Then,  $A_1=A_2=0$ is the unique solution of the system \eqref{ampliec0} for each $c \in\{\Im c\geq0\}^*$.
\end{lemma}
\begin{proof}
Clearly, if $\mathbf{r}\in span\{\wun',\wdo'\}$ solves $B_\gamma \mathbf{r}=0$, then \eqref{key1a} implies $\|\mathbf{r}\|\leq0$. By properties of the norm, the unique vector with this property is $\mathbf{r}=\zer$. But since $\wun',\wdo'$ are linearly independent, \eqref{facto1} implies that $A_1=A_2=0$ is the unique solution.
\end{proof}
Thus, the main purpose here is to derive the inequality \eqref{key1a}. The idea is to obtain a positive definite  quadratic form, from which \eqref{key1a} follows directly. The assumption $\Re\gamma_j=\varepsilon_j<0$, $j=1,2$ plays an essential role in the construction. Another crucial point shall be  the sign  of the quadratic form $$\mathbf{r}\to\mathbf{r}^* \cS'_2\mathbf{r}$$ 
restricted to $span\{\wun',\wdo'\}$, where $\cS_2'$  is the symmetric $4\times 4$ matrix obtained by eliminating both the first row and column from matrix $\cS_2$ in \eqref{sys2or}. For the sake of simplicity, we drop the subscript in the $\C^n$ norms.

\begin{theo}
\label{maint11}
If  $\gamma_1,\gamma_2$ are complex constants such that $\Re\gamma_j<0$, $j=1,2$, then there exist a  positive constant, $\epsilon>0$ (depending on $\gamma_j$ and $c_j$, $j=1,2$) such that 
\begin{equation}\label{ineq1}
\beta_0\|B_\gamma \mathbf{r}\|^2> \epsilon\|\mathbf{r}\|^2-\mathbf{r}^* \cS'_2\mathbf{r},
\end{equation}
for all vectors $\mathbf{r}\in\C^4\setminus\{\zer\}$, where
$$\beta_0:=-\dfrac{1}{2\rho\Re\gamma_1}-\dfrac{1}{2\rho\Re\gamma_2}>0.$$
\end{theo}
\begin{proof}
Let $\beta$ and $\epsilon$ be positive constants. What we must prove amounts to show that the function
\begin{equation}
G(\mathbf{r}):=\beta\big\|B_\gamma \mathbf{r}\big\|^2-\epsilon\|\mathbf{r}\|^2+\mathbf{r}^* \cS'_2\mathbf{r},
\end{equation}
is positive for $\beta=\beta_0$ and some $\varepsilon>0$. Since $\|B_\gamma \mathbf{r}\|^2=\mathbf{r}^*B^{*}_\gamma B_\gamma \mathbf{r}$ and $\|\mathbf{r}\|^2=\mathbf{r}^*\Idfo \mathbf{r}$,  $G$ is actually a hermitian quadratic form $G(\mathbf{r})=\mathbf{r}^*\M\mathbf{r}$, where $\M$ is the hermitian matrix  given by
\[
\M:=\beta B^{*}_\gamma B_\gamma-\epsilon\Idfo+\cS'_2.
\]
In this fashion, we just have to prove that $\M$ is positive definite for some $\epsilon>0$, that is,  all of its eigenvalues are positive.
 Observe first that the eigenvalues of $\M$ have  the form $\lambda-\epsilon$ , where $\lambda$ is an eigenvalue of the matrix 
\begin{equation}
\beta B^{*}_\gamma B_\gamma+\cS'_2=\begin{pmatrix} \beta c_2^2\rho^2 & 0 & c_2(1+\beta\gamma_1\rho) & 0\\
0 & \beta c_1^2\rho^2 & 0 & c_1(1+\beta\gamma_2\rho)\\
c_2(1+\beta\overline{\gamma}_1\rho) & 0 & \beta|\gamma_1|^2 & 0\\
0 & c_1(1+\beta\overline{\gamma}_2\rho) & 0 & \beta|\gamma_2|^2
\end{pmatrix}.
\end{equation}
This matrix has four  eigenvalues given by
\begin{equation}\label{eigen56}
\begin{aligned}
\lambda^{\pm}_1=&\dfrac{1}{2}\Big(\beta(c_1^2\rho^2+|\gamma_2|^2)\pm\sqrt{d_1}\Big),\\
\lambda^{\pm}_2=&\dfrac{1}{2}\Big(\beta(c_2^2\rho^2+|\gamma_1|^2)\pm\sqrt{d_2}\Big),
\end{aligned}
\end{equation}
where
\begin{equation}\label{discr12}
\begin{aligned}
d_1=&\beta^2\big(c_1^2\rho^2+|\gamma_2|^2\big)^2+4c_1^2\big(1+2\beta\rho\Re\gamma_2\big),\\
d_2=&\beta^2\big(c_2^2\rho^2+|\gamma_1|^2\big)^2+4c_2^2\big(1+2\beta\rho\Re\gamma_1\big).
\end{aligned}
\end{equation}
It is clear that $\lambda^{+}_1,\lambda^{+}_2$ are positive for all $\beta>0$. For the remaining eigenvalues,  we  make
$$\beta=\beta_0=-\dfrac{1}{2\rho\Re\gamma_1}-\dfrac{1}{2\rho\Re\gamma_2}.$$
Then a straighforward calculation shows that
$$1+2\rho\beta_0\Re\gamma_2=-\dfrac{\Re\gamma_2}{\Re\gamma_1}<0.$$
Therefore, $d_1<\beta_0^2\big(c_1^2\rho^2+|\gamma_2|^2\big)^2$, which implies 
\[
2\lambda^{-}_1=\beta_0(c_1^2\rho^2+|\gamma_2|^2)-\sqrt{d_1}>0,
\]
that is, $\lambda^{-}_1>0$. Analogously, we conclude $\lambda^{-}_2>0$. If we choose 
$$\epsilon=\min\{\lambda^{-}_1,\lambda^{-}_2\}>0,$$
then all the eigenvalues of matrix  $\M$, namely $\lambda^{\pm}_1-\epsilon,\lambda^{\pm}_2-\epsilon$ are positive, so $\M$ is positive definite.
\end{proof}
Observe that the inequality \eqref{key1a} easily follows  from \eqref{ineq1} if $\mathbf{r}^* \cS'_2\mathbf{r}\leq0$. That is, the quadratic form  $\mathbf{r}\to\mathbf{r}^* \cS'_2\mathbf{r}$ restricted to the linear space $span\{\wun',\wdo'\}$ is non-positive for all $c\in \{\Im c\geq0\}^*$. To prove that, we can try to compute the associated matrix of the restricted quadratic form and show the non-positiveness by elementary tools.
To do this, note that any vector  $\mathbf{r}\in span\{\wun',\wdo'\}$ has the form $A_1\wun'+\wdo'A_2$ with $A_1,A_2\in\C$, that in matricial form can be written as
\[
A_1\wun'+\wdo'A_2=\cA\begin{pmatrix}A_1 \\ A_2
\end{pmatrix}=\cA\mathbf{a},
\]
where $\cA$  denote the $4\times 2$ complex matrix whose columns are  $\wun',\wdo'$ and $\mathbf{a}=(A_1,A_2)^{\top}\in\C^2$ (column vector). Therefore, for all $\mathbf{r}\in span\{\wun',\wdo'\}$ we have
\begin{equation}\label{quadr39}
\begin{split}
\mathbf{r}^*\cS'_2\mathbf{r}&=(A_1\wun'+\wdo'A_2)^*\cS'_2(A_1\wun'+\wdo'A_2)\\
&=\mathbf{a}^*\Big(\cA^*\cS'_2\cA\Big)\mathbf{a}.
\end{split}
\end{equation}
That is, the quadratic form $\mathbf{r}^*\cS'_2\mathbf{r}$ restricted to the two dimensional linear space $span\{\wun',\wdo'\}$ is equivalent to a non restricted quadratic form defined on $\C^2$ with associated hermitian matrix $\cA^*\cS'_2\cA$. A straighforward calculation gives
\begin{equation}
\begin{aligned}
\cA^*\cS'_2\cA=&2k^2c_2^2\ii\begin{pmatrix} c\overline{b_1}-b_1\overline{c} & \Big(\overline{b}_1+\dfrac{1}{2b_2}\Big)(c-\overline{c})\vspace{.2cm}\\ \Big(b_1+\dfrac{1}{2\overline{b}_2}\Big)(c-\overline{c}) & \dfrac{c}{b_2}-\dfrac{\overline{c}}{\overline{b}_2} \end{pmatrix}\\  &+ k^2c_2^2\ii\begin{pmatrix} b_1 c(1+\overline{b_2}^2)-\overline{b}_1\overline{c} (1+b_2^2) & \dfrac{1}{b_2}\Big(c\overline{b}_2^2-\overline{c}b_2^2\Big)\vspace{.2cm}\\ \dfrac{1}{\overline{b}_2}\Big(c\overline{b}_2^2-\overline{c}b_2^2\Big) & c\Big(\overline{b}_2+\dfrac{1}{\overline{b}_2}\Big)-\overline{c} \Big(b_2+\dfrac{1}{b_2}\Big) \end{pmatrix}.
\end{aligned}
\end{equation}
Note that each component of the matrices above have the form $z-\overline{z}$, $z\in\C$; if $z\in\R$, $z-\overline{z}=0$. This is precisely the situation if we assume $c\in (-c_2,c_2)\setminus\{0\}$ because this implies $b_1,b_2\in\R$, so all components of the matrix above are reals and therefore zero, that is $\cA^*\cS'_2\cA=\zer_{2\times 2}$. In view of \eqref{quadr39}, we have proved that 
$$\mathbf{r}^*\cS'_2\mathbf{r}=0,\; \forall\mathbf{r}\in span\{\wun',\wdo'\}.$$
Substituting the expression hereabove into \eqref{ineq1} yields the desired inequality \eqref{key1a} for all  $c\in (-c_2,c_2)\setminus\{0\}$. 
 To extend the latter conclusion to the whole complex half-plane $\{\Im c\geq0\}^*$, we need to show that matrix $\cA^*\cS'_2\cA$ is non-positive definite for all $c$ with $\Im c>0$. However, the usual criteria to prove it, such as the negative sign of the eigenvalues or the negative sign of the principal  minors are impossible to perform, or at least hardly realizable in practice, given the intricate expression of matrix $\cA^*\cS'_2\cA$. An alternative method is needed to prove the positivity or non-negativity of that matrix. It is based on the following lemma due to Serre \cite{serr2}, where  the symmetric system \eqref{sys2or} is  fundamental.


\begin{lemma}
All solution $\mathbf{w}=\mathbf{w}(x_1,x_2,t)$ of \eqref{sys2or} satisfies
\begin{equation}\label{ener2}
\dfrac{\partial |\mathbf{w}|^2}{\partial t}=\dfrac{\partial}{\partial x_1}(\mathbf{w}^{*}\cS_1\mathbf{w})+ \dfrac{\partial }{\partial x_2}(\mathbf{w}^{*}\cS_2\mathbf{w}).
\end{equation}
\end{lemma}
\begin{proof}
Multiplying \eqref{sys2or} to the left by $2\mathbf{w}^*$ and taking the real part yield
\begin{equation}\label{ener4}
2\Re\Big(\mathbf{w}^*\dfrac{\partial \mathbf{w}}{\partial t}\Big)=2\Re\Big(\mathbf{w}^*\cS_1\dfrac{\partial \mathbf{w}}{\partial x_1}\Big)+ 2\Re\Big(\mathbf{w}^*\cS_2\dfrac{\partial \mathbf{w}}{\partial x_2}\Big).
\end{equation}
Since $\cS_j$, $j=1,2$ are real symmetric constant matrices, trivially we have  $(\cS_j)^*=\cS_j$. Therefore, by the usage of  the properties $\overline{\mathbf{v}^{*}\,\mathbf{w}}=\mathbf{w}^{*}\,\mathbf{v}$, where $\mathbf{v},\mathbf{w}\in\C^n$ and  $2\Re z=z+\overline{z}$, where $z\in\C$, we obtain
\begin{equation}\label{enerk}
\begin{split}
2\Re\Big(\mathbf{w}^*\cS_j\dfrac{\partial \mathbf{w}}{\partial x_j}\Big)&=\overline{\mathbf{w}^*\cS_j\dfrac{\partial \mathbf{w}}{\partial x_j}}+\mathbf{w}^*\cS_j\dfrac{\partial \mathbf{w}}{\partial x_j}\\
&=\left(\cS_j\dfrac{\partial \mathbf{w}}{\partial x_j}\right)^*\mathbf{w}+\mathbf{w}^*\cS_j\dfrac{\partial \mathbf{w}}{\partial x_j}\\
&=\dfrac{\partial \mathbf{w}^{*}}{\partial x_j}\cS_j\mathbf{w}+\mathbf{w}^*\dfrac{\partial }{\partial x_j}\cS_j\mathbf{w}.
\end{split}
\end{equation}
Finally, the product rule for the derivatives gives us
\[
2\Re\Big(\mathbf{w}^*\cS_j\dfrac{\partial \mathbf{w}}{\partial x_j}\Big)=\dfrac{\partial}{\partial x_j}(\mathbf{w}^{*}\cS_j\mathbf{w}),\:j=1,2.
\]
The above procedure is valid if we replace $\cS_j$ by $\Idc$ and consider the temporal derivative. Hence, we also have 
\[
2\Re\Big(\mathbf{w}^*\dfrac{\partial \mathbf{w}}{\partial t}\Big)=\dfrac{\partial}{\partial t}(\mathbf{w}^{*}\mathbf{w})=\dfrac{\partial |\mathbf{w}|^2}{\partial t}.
\]
Susbtituting back the latter two expression into \eqref{ener4} yields the result.
\end{proof}

Now, we can prove the principal result of this section.
\begin{theo}\label{nonva1}
Let $\gamma_1,\gamma_2$ be complex constants with negative real part. Then the secular equation \eqref{finsec12} (associated to the PBC) does not have roots in the upper complex half-plane $\{\Im c\geq0\}^*$.
\end{theo}
\begin{proof}
According Lemma \ref{keylemma}, the result follows from inequality \eqref{key1a}. This can be obtained from \eqref{ineq1} if  we  prove that for each $c\in \{\Im c\geq0\}^*$,  $\mathbf{r}^* \cS'_2\mathbf{r}\leq0$ for all $\mathbf{r}\in span\{\wun',\wdo'\}$.  This shall follow from the straighforward application of the identity \eqref{ener2} to the surface  wave solution \eqref{moddef56}.
Indeed, let us define first
\[
\mathbf{p}(x_2):=A_1\textit{\LARGE e}^{-kb_1 x_2}\wun+A_2\textit{\LARGE e}^{-kb_2 x_2}\wdo.
\]
Thus, the surface wave solution \eqref{moddef56} takes the form $\mathbf{w}=\mathbf{p}(x_2)\textit{\LARGE e}^{k\ii(x_1-ct)}$. A straightforward computation gives
\[
|\mathbf{w}|^2=\textit{\LARGE e}^{2k t\Im c}|\mathbf{p}|^2,\quad \mathbf{w}^{*}\cS_j\mathbf{w}=\textit{\LARGE e}^{2k t\Im c}\mathbf{p}^{*}\cS_j\mathbf{p},\; j=1,2.
\]
Note that $\mathbf{w}^{*}\cS_1\mathbf{w}$ does not depend on $x_1$, so substituting the expressions obtained hereabove into \eqref{ener2} and simplifying yield
\begin{equation}\label{posi5}
\dfrac{\partial }{\partial x_2}(\mathbf{p}^{*}\cS_2\mathbf{p})= 2k|\mathbf{p}|^2\Im c,
\end{equation}
which is non negative as $k>0$ and $c$ lies on the upper complex half-plane ($\{\Im c\geq0\}^*$). Thus, $\mathbf{p}^{*}\cS_2\mathbf{p}$ is a non-decreasing function of $x_2\in(0,+\infty)$. So, in particular
$$\mathbf{p}^{*}(x_2)\cS_2\mathbf{p}(x_2)\geq\mathbf{p}^{*}(0)\cS_2\mathbf{p}(0),\:\text{for all}\: x_2\geq0.$$ Now, suppose by the way of contradiction that $\mathbf{p}^{*}(0)\cS_2\mathbf{p}(0)>0$. Then, the  Cauchy-Schwarz inequality applied to $\mathbf{p}^{*}\cS_2\mathbf{p}$ and the non-decreasing property on $[0,+\infty)$ yield
$$\|\cS_2\|\|\mathbf{p}\|^2\geq\mathbf{p}^{*}(x_2)\cS_2\mathbf{p}(x_2)\geq\mathbf{p}^{*}(0)\cS_2\mathbf{p}(0)>0,\:\text{for all}\: x_2\geq0,$$ 
which means that $\|\mathbf{p}\|$ does not decrease to zero when $x_2\to\infty$, provided that $\|\cS_2\|$ is a constant. This contradicts the decay condition \eqref{decay1}, which characterizes a surface wave solution of Rayleigh type.
 Thus, it is necessary that
\begin{equation}\label{posi6}
\mathbf{p}^{*}(0)\cS_2\mathbf{p}(0)\leq0.
\end{equation}
Note that $\mathbf{p}(0)$ is any linear combination of  the vectors $\wun,\wdo$ and since  the first column and row of $\cS_2$ are full of zeros, the inner product in \eqref{posi6} coincides with $\mathbf{r}^*\cS'_2\mathbf{r}$ for all $\mathbf{r}\in span\{\wun',\wdo'\}$. Therefore we have demonstrated that once $\Im c\geq0$, $c\neq0$ then
\begin{equation}
\label{alprov2}
\mathbf{r}^*\cS'_2\mathbf{r}\leq0,\; \forall\mathbf{r}\in span\{\wun',\wdo'\}.
\end{equation}
Combining \eqref{ineq1} with the inequality here above yields
\[
\begin{aligned}
\beta_0\|B_\gamma \mathbf{r}\|^2&>\epsilon\|\mathbf{r}\|^2-\mathbf{r}^* \cS'_2\mathbf{r}\\
&\geq \epsilon\|\mathbf{r}\|^2.
\end{aligned}
\]
which is the key inequality \eqref{key1a}. Lemma \ref{keylemma} implies that the unique solution of the linear system \eqref{facto1} (or \eqref{ampliec0}) is the trivial one $A_1=A_2=0$. That is, the secular equation does not vanish along $\{\Im c\geq0\}^*$. 
\end{proof}

\begin{remark}
Notice that the last theorem does not hold for $\varepsilon_j=\Re\gamma_j=0$, $j=1,2$, because  Inequality \eqref{ineq1}, which was crucial for the proof, ceases to be  valid as $\beta_0$ becomes undefined when $\Re\gamma_1=\Re\gamma_2=0$. In other words, Theorem \ref{nonva1} does not hold for the impedance boundary condition; whence the need of considering the PBC \eqref{bouny0} first.
Within the framework of the elasticity theory, a suface wave of Rayleigh type is often associated to real solution of the secular equation. Hence, according to  Theorem  \ref{nonva1}, one might infer in particular that  surface waves of Rayleigh type are impossible under the PBC.  However, there might be roots with $\Im c<0$ that could be associated with surface waves, which, strangely enough would present exponential decay over time due to the negative sign of the imaginary part of the root. 
 Although scarce, some atypical theoretical findings regarding Rayleigh waves can be found in the material science literature. For instance, Kuznetsov \cite{Kuzn03,Kuzn02} showed that it is theoretically possible for some anisotropic elastic materials to exhibit properties of non-existence of the genuine Rayleigh waves or surface waves of non-Rayleigh type. Nevertheless, the author points out that the question whether there actually exist that kind of waves remains open. In the hyperbolic PDEs literature, Theorem \ref{nonva1} is related to the  uniform Kreiss-Lopatinski\u{\i} condition which  implies the well-possedness of the boundary value problem \eqref{sys2or}-\eqref{bouny} (see \cite{Morando05,BS}).

\end{remark}


To observe the behavior of the secular equation for some particular values of the parameters $\gamma_1,\gamma_2$, consider an elastic half-space with density $\rho=1$, Young's modulus $E=1.86$  and Poisson's ratio $\nu=0.16$ (that is, $\mu=0.8$ and $\lambda=0.4$). Figure \ref{char1} shows some plots of the norm of the Rayleigh function $\cR(c,\gamma_1,\gamma_2)$  as a function of $c$ on the positive real axis. Observe that as the negative  parameters $\gamma_1,\gamma_2$ get closer to zero, a root of  $|\cR|$ on the interval $(0,c_2)$ start to appear. This corresponds to the velocity of the well-known Rayleigh wave (when $\gamma_1=\gamma_2=0$). Moreover, Figure \ref{fig99}   shows plots of the same function with the same parameter values, but with $c$ lying in the upper complex half-plane $\Im c>0$. It is easy to verify that there are no roots in that region.

\begin{figure}[h]
\begin{center}
\includegraphics[scale=.8, clip=true]{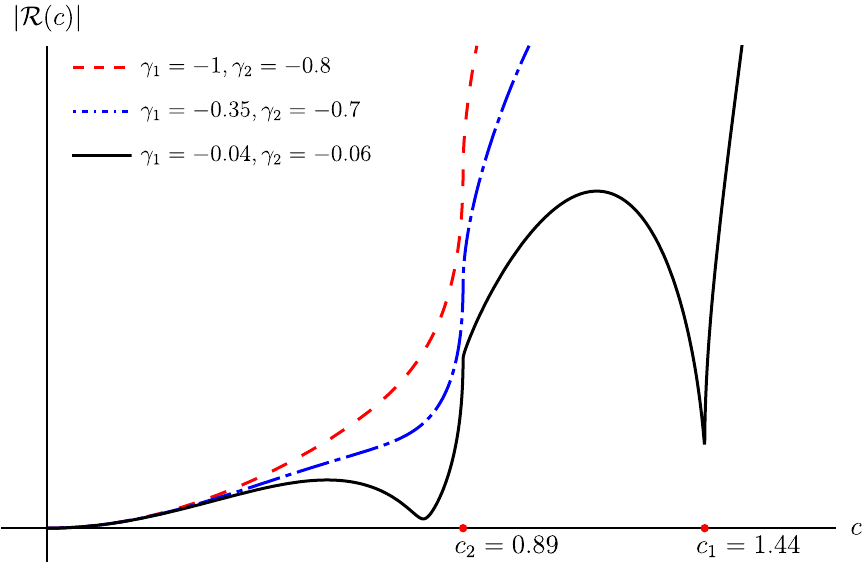}
\end{center}
\caption{Plots of $|\cR(c;\gamma_1,\gamma_2)|$ as a function of $c\in[0,\infty)$ for $\rho=1$, $E=1.86$, $\nu=0.16$ and some negative values of the boundary parameters $\gamma_1,\gamma_2$. (Color online)}\label{char1}
\end{figure}


\begin{figure}
\begin{center}
\subfigure[ $\gamma_1 = -0.35$, $\gamma_2 = -0.7$]{\label{figa}\includegraphics[scale=.5, clip=true]{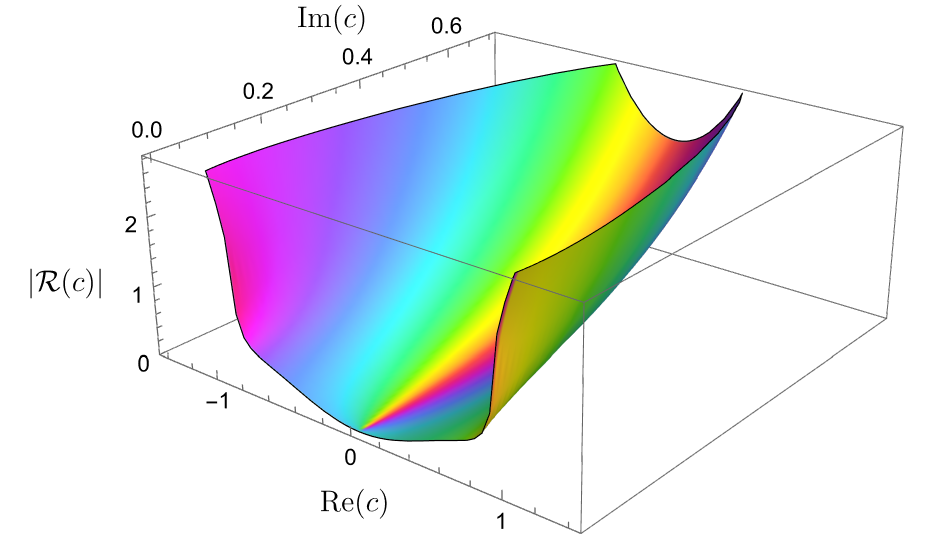}
\includegraphics[scale=.8, clip=true]{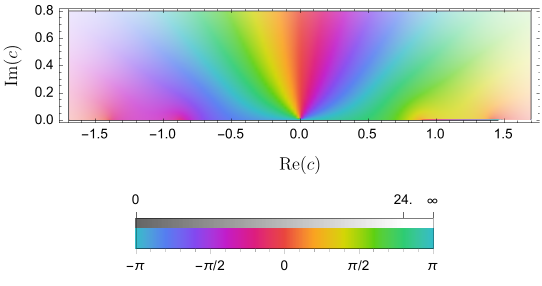}}
\subfigure[ $\gamma_1 = -0.04$, $\gamma_2 = -0.06$]{\label{figb}\includegraphics[scale=.5, clip=true]{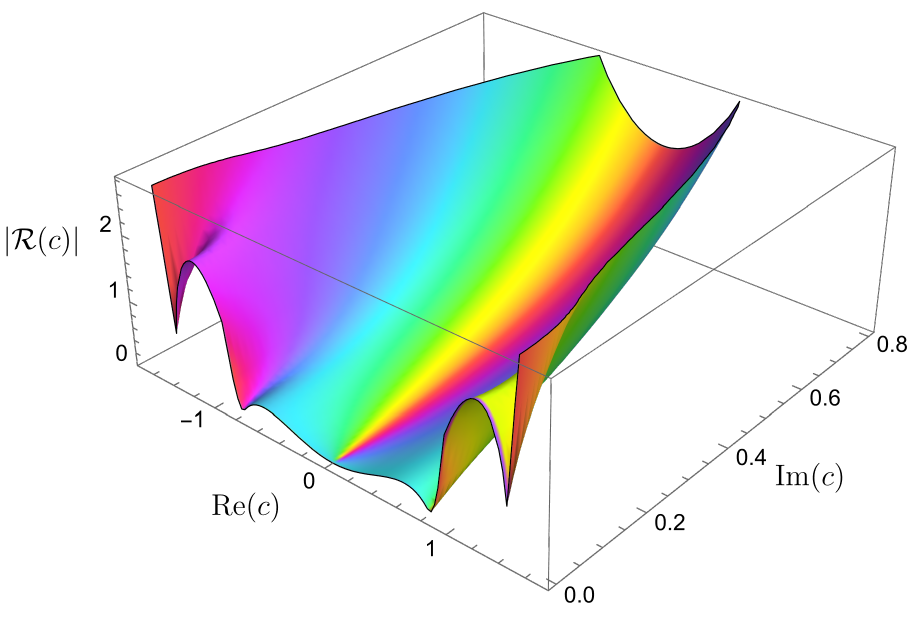}
\includegraphics[scale=.8, clip=true]{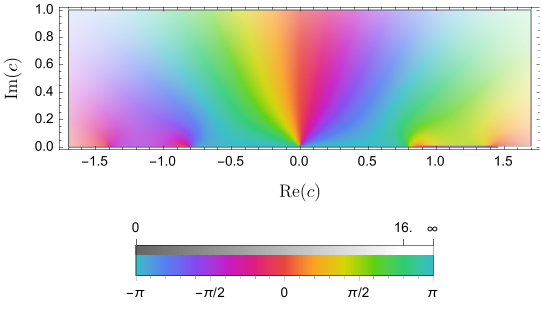}}
\end{center}
\caption{Complex plot (in 3D, left, and contour, right) of the norm of Rayleigh wave function \eqref{finsec12}  as function of $c\in \{\Im c\geq0\}$, with elastic parameters values, $E = 1.86$, $\nu=0.16$ and  boundary parameter  values $\gamma_1 = -0.35$, $\gamma_2 = -0.7$ (panel (a)) and $\gamma_1 = -0.04$, $\gamma_2 = -0.06$ (panel (b)). The color mapping legend shows the modulus $|\cR| \in (0,\infty)$ from dark to light tones of color and the phase from light blue ($\text{arg}(c) = -\pi$) to green ($\text{arg}(c) = \pi$). (Color online) }\label{fig99}
\end{figure}

\newpage
\section{The secular equation for  impedance boundary conditions}
\label{imped678}
In this section we let $\varepsilon_1,\varepsilon_2$ go to zero in \eqref{finsec12} and investigate if the non-vanishing property demonstrated in Theorem \ref{nonva1} remains valid in the limit. That is, the secular equation  for surface waves with impedance boundary conditions \eqref{finsec50}. We proceed in a classical way, approximating  \eqref{finsec50} by a sequence of secular equations with PBC for which Theorem \eqref{nonva1} guarantees the non-vanishing property along $\{\Im c\geq0\}^*$. When taking the limit in the sequence, we shall see that the non-vanishing property remains valid only on the open complex half-plane $\{\Im c>0\}$, not along the real axis. 
 Let us define
\begin{equation}\label{finsec5}
\begin{split}
f(c):&=\cR(c;Z_1\ii,Z_2\ii)\\&=\left(2-\dfrac{c^2}{c_2^2}\right)^2 -4\sqrt{1-\dfrac{c^2}{c_2^2}}\sqrt{1-\dfrac{c^2}{c_1^2}}+\dfrac{c^3}{\mu c_2^2}\left(Z_1\sqrt{1-\dfrac{c^2}{c_2^2}}+Z_2 \sqrt{1-\dfrac{ c^2}{c_1^2}}\right)\\ &+c^2\frac{Z_1 Z_2}{\mu^2}\left(\sqrt{1-\dfrac{c^2}{c_2^2}}\sqrt{1-\dfrac{c^2}{c_1^2}}-1\right).
\end{split}
\end{equation}
We can  write  \eqref{finsec5} as the limit of a sequence of functions of the form $\cR$ in \eqref{finsec12} (associated to the PBC) by setting the boundary parameters $\gamma_j(n):=-\frac{1}{n}+Z_j\ii$, $j=1,2$ (that is, $\varepsilon_j(n):=-\frac{1}{n}$) for each $n\in\Z^+$. That is, 
\begin{equation}\label{limit1}
f(c)=\cR\big(c;Z_1\ii,Z_2\ii\big)=\lim_{n\to\infty}\cR\Big(c;\gamma_1(n),\gamma_2(n)\Big).
\end{equation}
Since we trivially have 
\[
\Re\big(\gamma_j(n)\big)=-\tfrac{1}{n}<0,\; j=1,2,
\]
Theorem \ref{nonva1} implies that each element of the sequence 
\begin{equation}\label{seque67}
f_n(c):=\cR\big(c;\gamma_1(n),\gamma_2(n)\big),\;n\in\Z^+
\end{equation}
does not vanish on $\{\Im c\geq0\}^*$ for each $n\in \Z^+$. One would infer that this non-vanishing property remains valid as $n$ goes to infinity, but it is widely known that properties such as continuity or integrability of a given  sequence of functions may be lost at infinity  (see, e.g., \cite{vallejo2010}). However, there is a classical result from complex analysis that guarantees the non-vanishing property of the limit function $f=\cR$, but only on the upper complex half-plane, not along the real axis. This is  Hurwitz's theorem (see  \cite{ahlfors79}).
\begin{theo}[\emph{Hurwitz's theorem}]
\label{Hurw}
Let $\Omega\in\C$  an open connected set and suppose the sequence of analytic functions $f_n:\Omega\to\C$, $n\in\Z^+$ converges   to $f$ uniformly on every compact subset of $\Omega$. If each  $f_n$ never vanishes on $\Omega$, then either $f$ is identically zero or $f$ never vanishes on $\Omega$.
\end{theo}
A compact set in $\C$ (or $\R^2$) is a bounded set containing its boundary points. This theorem has been applied for instance to count zeros of holomorphic (or analytic) functions inside of open connected sets, see for instance \cite{asmar18}.  We shall apply Hurwitz's theorem to the sequence of functions $f_n$, $n\in\Z^+$ defined in \eqref{seque67}.
Since the theorem requires that the sequence be defined on an open connected set, we assume that each $f_n$ in \eqref{seque67} is defined on $\Omega=\{c\in\C:\Im c>0\}$, which meets the condition. Recall that it was assumed that each of the square roots in \eqref{finsec5} is the principal branch, hence they are  holomorphic on the entire complex plane except for the   branch cut located along the real axis. As a result, $f_n$ is in particular holomorphic on $\Omega=\{c\in\C:\Im>0\}$. It remains to verify the uniform convergence  of the sequence  $\{f_n\}$ on any compact set contained in $\Omega$. This is accomplished in the following lemma.
\begin{lemma}\label{convnormal}
For all fixed pair $Z_1, Z_2\in\R$, the sequence of functions $f_n$ defined in \eqref{seque67} converges uniformly to $f(c)=\cR\big(c;Z_1\ii,Z_2\ii\big)$ on any compact set $K\subset\Omega$.
\end{lemma}
\begin{proof}
Let $K$ be an arbitrary compact subset of $\Omega$. To prove the uniform convergence of $f_n$ to $f$ on $K$, we have to show that (see \cite{1henrici98}, section \S 2.1, Chapter 1): 
\begin{equation}\label{unif1}
\lim\limits_{n\to\infty}\sup_{c\in K}|f_n(c)-f(c)|=0.
\end{equation}
To show this, let us define the functions
$$g_1(c):=-\dfrac{c^3\ii}{\mu c_2^2}\sqrt{1-\dfrac{c^2}{c_2^2}},\quad g_2(\tau):=-\dfrac{c^3\ii}{\mu c_2^2} \sqrt{1-\dfrac{ c^2}{c_1^2}},\quad g_3(c):=\frac{c^2}{\mu^2}\left(1-\sqrt{1-\dfrac{c^2}{c_2^2}}\sqrt{1-\dfrac{c^2}{c_1^2}}\right).$$
The Rayleigh function in \eqref{finsec12} can be written as
\[
\cR(c;\gamma_1,\gamma_2)=\left(2-\dfrac{c^2}{c_2^2}\right)^2 -4\sqrt{1-\dfrac{c^2}{c_2^2}}\sqrt{1-\dfrac{c^2}{c_1^2}}+\gamma_1 g_1(c)+\gamma_2g_2(c)+\gamma_1\gamma_2g_3(c).
\]
Thus, as $f_n$  and $f$  are defined in terms of the function $\cR$ (see, \eqref{seque67} and \eqref{finsec5}),  straighforward calculation gives
\begin{equation}\label{ineqp9}
\begin{split}
|f_n(c)-f(c)|&=\big|-\tfrac{1}{n}g_1(c)-\tfrac{1}{n}g_2(c)+\big((-\tfrac{1}{n}+Z_1\ii)(-\tfrac{1}{n}+Z_2\ii)-(Z_1\ii)(Z_2\ii)\big)g_3(c)\big|\\
&\leq \tfrac{1}{n}|g_1(c)|+\tfrac{1}{n}|g_2(c)|+\big|(-\tfrac{1}{n}+Z_1\ii)(-\tfrac{1}{n}+Z_2\ii)-(Z_1\ii)(Z_2\ii)\big| |g_3(c)|.
\end{split}
\end{equation}
Note that $g_1,g_2,g_3$ are holomorphic in $\Omega$. Hence taking the complex norm $|g_1|,|g_2|,|g_3|$,  produces real valued continuous functions defined on the compact set $K\subset\Omega$. Each continuous real valued function on a compact set of $\C=\R^2$ attains a maximum on that set. Thus, there are positive constants $m_1,m_2,m_3$ such that
\[
|g_1(c)|\leq m_1,\quad|g_2(c)|\leq m_2, \quad|g_3(c)|\leq m_3,
\] 
for all $c\in K$. Using this fact into \eqref{ineqp9} gives
\begin{equation}\label{supm45}
|f_n(c)-f(c)|\leq \tfrac{1}{n}m_1+\tfrac{1}{n}m_2+\big|(-\tfrac{1}{n}+Z_1\ii)(-\tfrac{1}{n}+Z_2\ii)-(Z_1\ii)(Z_2\ii)\big| m_3,
\end{equation}
for all $c\in K$. Note that the right hand side from the  above, namely,
$$r_n:=\tfrac{1}{n}m_1+\tfrac{1}{n}m_2+\big|(-\tfrac{1}{n}+Z_1\ii)(-\tfrac{1}{n}+Z_2\ii)-(Z_1\ii)(Z_2\ii)\big| m_3,$$
is an upper bound (independent of $c$) for the values of $|f_n(c)-f(c)|$ on $K$. Thus, since the supremum is the least upper bound, we trivially have from \eqref{supm45} that
\[
0\leq\sup_{c\in K}|f_n(c)-f(c)|\leq  r_n.
\]
Given that  $r_n$ tends to  zero as  $n$ goes to $\infty$, we obtain  \eqref{unif1} from the inequality hereabove by applying the squeeze theorem for sequences. That is, the convergence $f_n\to f$ is uniform on $K$. Since $K$ is an arbitrary compact set in $\Omega$, we have proof that the convergence is uniform on any compact set of $\Omega$. 
\end{proof}

 Now we can proof the principal result
\begin{theo}\label{godosecu1}
Let $\lambda,\mu$ as in \eqref{lame44} and $Z_1,Z_2$ be real constants. The secular equation with impedance boundary condition (see  \eqref{finsec5})
\begin{equation}\label{finfin}
f(c)=\cR(c;Z_1\ii,Z_2\ii)=0,
\end{equation}
 has no roots  outside of the real axis. 
\end{theo}
\begin{proof}
According to Lemma \ref{convnormal}, the sequence \eqref{seque67} converges uniformly on any compact set of the open connected set $\Omega=\{c\in\C:\Im c>0\}$. Since the perturbed parameters $\gamma_j(n)=-\frac{1}{n}+\ii Z_j$, $j=1,2$ have negative real parts, Theorem \eqref{nonva1} ensures that each holomorphic function $f_n$ never vanishes on $\Omega$ so Hurwitz's theorem applies and then the limit function $f(c)=\cR(c;Z_1\ii,Z_2\ii)$ never vanishes on $\Omega$ provided it is not identically zero for all $Z_j\in\R$, $j=1,2$. 
Moreover, the symmetric property \eqref{unif89} implies that the secular equation neither has roots along the open lower complex half-plane $\{\Im c<0\}$.
\end{proof}

\begin{remark}
Theorem \ref{godosecu1} is consistent with the results obtained in the literature for particular cases of the full impedance boundary conditions under consideration. In the stress-free case ($Z_1=Z_2=0$) and the tangential case ($Z_1\in\R, Z_2=0$), the associated secular equation does not have roots off the real axis, and there is always a unique real root on the interval $(0, c_2)$.
 The case of normal impedance boundary conditions ($Z_1\in\R, Z_2=0$) is quite interesting since it was demonstrated  (see Theorems 1 and 3 in \cite{Pham21}) that the existence of a surface wave of Rayleigh type  is lost when the impedance parameter $Z_2$ decreases  beyond of a critical value. This  suggest that the non-existence of surface waves or even the existence of solutions of infinite energy could occur when considering the general case with both non-zero impedance parameters.  Theorem \ref{godosecu1} at least rules out the second scenario.  
   In the context of initial boundary value problems of first order hyperbolic systems of PDEs defined on the half-space, this non-vanishing property of the secular equation in the upper complex half-plane is equivalent to the (weak) Kreiss-Lopatinski\u{\i} condition \cite{BS,BRSZ,Kre70,Hig01,serr2}, a necessary condition for the well-posedness of boundary value problems of the form \eqref{sys2or}-\eqref{prefacto}.

\end{remark}

\section{Discussion}
\label{secdisc}
In this paper, we have presented an alternative method for addressing the secular equation of surface waves propagating in an elastic isotropic half-space subjected to boundary conditions of  impedance type. We first consider a boundary condition (PBC) that can be viewed as a perturbed version of the impedance boundary condition proposed by Godoy et al. \cite{Godoy1}.  The method consist in studying the associated secular equation  indirectly through the algebraic homogeneous linear system, whose determinant gives rise the secular equation itself.
 This alternative approach becomes feasible thanks to both, the symmetric first order version of the isotropic equations and the representation of the algebraic  homogeneous linear system (whose determinant is the secular equation)  as a constant matrix (related to the boundary condition) acting over a linear space spanned by two vectors associated to the surface normal modes.  In this fashion,  Inequality \eqref{ineq1} reduces the analysis to  determine the sign of a quadratic form to conclude that the secular equation (associated to the PBC) does not vanish neither along the upper complex half-plane nor the real axis (except for the trivial root $c=0$). 
 By the use of a classic approximation technique,  we aim to extend this non-vanishing property to the secular equation with full impedance boundary condition, finally proving that it does not have complex root outside the real axis for arbitrary impedance parameters. This is a necessary condition for the well-posedness of the boundary value problem, and thus crucial for the model to explain surface  wave propagation.  Since the behavior of the secular equation off the real axis has been established, simpler techniques can be used to analyze it along the real axis in order to demonstrate the existence and uniqueness of a surface wave. An appealing  example in this direction is the approach implemented by Godoy et al. \cite{Godoy1}, which, though intricate, relies on elements of basic calculus. 
Moreover, the fact that the secular equation with full impedance boundary condition can only  have real roots might indicate further simplifications when applying the complex function method to derive an exact formula  for the surface wave speed, in cases where it exists.
 It is worth noting that our approach seems to be extendable to the study of impedance boundary conditions for some anisotropic elastic solids. The reason is that one of the key ingredients, namely a symmetric first order versions of the equation of motion, is possible due to the quadratic form of the strain-energy function (see, \cite{Morando05,aki20}). Conversely,  another crucial element, namely Inequality \eqref{ineq1}  might be challenging to derive for the anisotropic case in the same way as in the proof of Theorem \eqref{maint11}. However, its existence might follow by certain results within the theory of hyperbolic systems of PDEs or advanced linear algebra, provided that it is essentially a quadratic form.

\section*{Acknowledgements}
The author is warmly grateful to Ram\'on G. Plaza  for many stimulating conversations. I also thank Federico J. Sabina, whose comments and suggestions improve the manuscript. 
 This work was supported by the National Science and Technology Council (CONAHCyT) of M\'exico under grant CF-2023-G-122.

\bibliographystyle{abbrvnat}
\bibliography{bibliography}

\end{document}